\documentclass[preprint,12pt,3p]{elsarticle}

\usepackage{amssymb}

\usepackage{amsthm}
\usepackage{graphicx}
\usepackage{amsmath}
\usepackage{hyperref}

\usepackage{pgf,tikz}

\usetikzlibrary{arrows}
\usetikzlibrary{shapes}
\usepackage{multirow}
\usepackage{latexsym}
\usepackage{setspace}
\PassOptionsToPackage{boxed,section}{algorithm}
\usepackage{enumerate}
\usepackage{amsmath}			
\usepackage{amssymb}			
\usepackage{url}
\usepackage{setspace}
\usepackage{tikz}
\usetikzlibrary{arrows}
\usetikzlibrary{shapes}
\usetikzlibrary{decorations.markings}
\usepackage{geometry}

\newtheorem{proposition}{\em Proposition}
\newtheorem{theorem}{\em Theorem}
\newtheorem{conjecture}{\em Conjecture}
\newtheorem{definition}{\em Definition}
\newtheorem{lemma}{\em Lemma}
\newtheorem{corollary}{\em Corollary}

\newtheorem{remark}{\em Remark}

\journal{Sample Journal}

\begin{document}

\begin{frontmatter}

\title{Normal edge-colorings of cubic graphs}

\author[label1]{Giuseppe Mazzuoccolo\corref{cor1}}
\address[label1]{Dipartimento di Informatica,
Universita degli Studi di Verona, Strada le Grazie 15, 37134 Verona, Italy}

\cortext[cor1]{Corresponding author}

\ead{giuseppe.mazzuoccolo@univr.it}

\author[label1]{Vahan Mkrtchyan}
\ead{vahanmkrtchyan2002@ysu.am}


\begin{abstract}
A normal $k$-edge-coloring of a cubic graph is an edge-coloring with $k$ colors having the additional property that when looking at the set of colors assigned to any edge $e$ and the four edges adjacent it, we have either exactly five distinct colors or exactly three distinct colors. We denote by $\chi'_{N}(G)$ the smallest $k$, for which $G$ admits a normal $k$-edge-coloring.
Normal $k$-edge-colorings were introduced by Jaeger in order to study his well-known Petersen Coloring Conjecture. More precisely, it is known that proving $\chi'_{N}(G)\leq 5$ for every bridgeless cubic graph is equivalent to proving Petersen Coloring Conjecture and then, among others, Cycle Double Cover Conjecture and Berge-Fulkerson Conjecture.   
Considering the larger class of all simple cubic graphs (not necessarily bridgeless), some interesting questions naturally arise. For instance, there exist simple cubic graphs, not bridgeless, with $\chi'_{N}(G)=7$. On the other hand, the known best general upper bound for $\chi'_{N}(G)$ was $9$. Here, we improve it by proving that $\chi'_{N}(G)\leq7$ for any simple cubic graph $G$, which is best possible. We obtain this result by proving the existence of specific nowhere zero $\mathbb{Z}_2^2$-flows in $4$-edge-connected graphs.
\end{abstract}


\end{frontmatter}

\section{Introduction}
\label{sec:intro}

The Petersen Coloring Conjecture is an outstanding conjecture in graph theory which asserts that the edge-set of every bridgeless cubic graph $G$ can be colored by using as set of colors the edge-set of the Petersen graph $P$ in such a way that adjacent edges of $G$ receive as colors adjacent edges of $P$.
The conjecture is well-known and it is largely considered hard to prove since it implies some other classical conjectures in the field such as Cycle Double Cover Conjecture and Berge-Fulkerson Conjecture (see \cite{Fulkerson,Jaeger1985,Zhang1997}).
Jaeger, in \cite{Jaeger1985}, introduced an equivalent formulation of the Petersen Coloring Conjecture. More precisely, he proved that a bridgeless cubic graph is a counterexample to this conjecture, if and only if, it does not admit a normal edge-coloring (see Definitions \ref{def:poorrich} and \ref{def:normal} in Section \ref{sec:intro}) with at most $5$ colors. We call normal chromatic index of $G$, denoted by $\chi'_{N}(G)$, the minimum number of colors in a normal edge-coloring of $G$. In this terms, Petersen Coloring Conjecture is equivalent to saying that every bridgeless cubic graph has normal chromatic index at most 5. As far as we know, the best known upper bound for an arbitrary bridgeless cubic graph is $7$ (see Theorem \ref{thm:7bridgeless}). 
A similar situation appears in the larger class of all simple cubic graphs (not necessarily bridgeless). Indeed, there exist examples of cubic graphs with normal chromatic index $7$, but the best known upper bound was $9$ (see \cite{Bilkova12}). This bound is obtained by a refinement of the proof used in \cite{Andersen1992} to show the existence of a strong edge-coloring of a cubic graph with 10 colors.
The upper bound for bridgeless cubic graphs is deduced by the $8$-flow Theorem of Jaeger. Following the same spirit, we approach the problem of finding a better upper bound for the class of all simple cubic graph by using flow theory.
In Section \ref{sec:aux}, we prove some technical lemmas which are refinements of some well-known statements in flow theory, such as the existence of a nowhere-zero 4-flows in graphs with two edge-disjoint spanning trees. Then, we use such results in Section \ref{sec:Main} to prove that every simple cubic graph has normal chromatic index at most 7. Due to the existence of examples where 7 colors are necessary, the proved upper bound is best possible.
Finally, we propose an Appendix where we present counterexamples for two possible natural stronger versions of our lemmas in Section \ref{sec:Main}, by proving that in some sense the results are optimal. 

Now, let us introduce the main definitions and notions used in the paper in some detail. Graphs considered in this paper are finite and undirected. They do not contain loops,
though they may contain parallel edges. We also consider pseudo-graphs, which may
contain both loops and parallel edges, and simple graphs, which contain neither loops nor parallel edges. As usual, a loop contributes to the degree of a vertex by two.



For a graph $G$ and a vertex $v$ let $\partial_{G}(v)$ be the set of edges of $G$ that are incident to $v$ in $G$. If $G$ is cubic and $F\subseteq E(G)$, then $F$ is a perfect matching of $G$ if and only if $E(G)\setminus E(F)$ is an edge-set of a $2$-factor of $G$. This perfect matching and $2$-factor are said to be complementary to each other.




Let $G$ and $H$ be two cubic graphs. If there is a mapping $\phi:E(G)\rightarrow E(H)$, such that for each $v\in V(G)$ there is $w\in V(H)$ such that $\phi(\partial_{G}(v)) = \partial_{H}(w)$, then $\phi$ is called an $H$-coloring of $G$. If $G$ admits an $H$-coloring, then we will write $H
\prec G$. It can be easily seen that if $H\prec G$ and $K\prec H$, then $K\prec G$. In other words, $\prec$ is a transitive relation defined on the set of cubic graphs.


    \begin{figure}[ht]
	\begin{center}
	\begin{tikzpicture}[style=thick]
\draw (18:2cm) -- (90:2cm) -- (162:2cm) -- (234:2cm) --
(306:2cm) -- cycle;
\draw (18:1cm) -- (162:1cm) -- (306:1cm) -- (90:1cm) --
(234:1cm) -- cycle;
\foreach \x in {18,90,162,234,306}{
\draw (\x:1cm) -- (\x:2cm);
\draw[fill=black] (\x:2cm) circle (2pt);
\draw[fill=black] (\x:1cm) circle (2pt);
}
\end{tikzpicture}
	\end{center}
	\caption{The graph $P_{10}$.}\label{fig:Petersen10}
\end{figure}
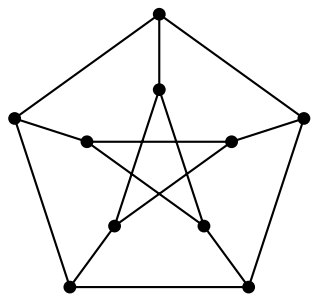

Let $P_{10}$ be the well-known Petersen graph (Figure \ref{fig:Petersen10}). The Petersen coloring conjecture of Jaeger states:
\begin{conjecture}\label{conj:P10conj} (Jaeger, 1988 \cite{Jaeger1988}) For any bridgeless cubic
graph $G$, one has $P_{10} \prec G$.
\end{conjecture}

Note that the Petersen graph is the only 2-edge-connected cubic graph that can color all bridgeless cubic graphs \cite{Mkrt2013}. The conjecture is difficult to prove, since it can be seen that it implies the
following two classical conjectures:
\begin{conjecture} (Berge-Fulkerson, 1972 \cite{Fulkerson,Seymour}) Any bridgeless
cubic graph $G$ contains six (not necessarily distinct) perfect matchings
$F_1, \ldots , F_6$ such that any edge of $G$ belongs to exactly two of them.
\end{conjecture}

\begin{conjecture}
((5, 2)-cycle-cover conjecture, \cite{Celmins1984,Preiss1981}) Any bridgeless
graph $G$ (not necessarily cubic) contains five even subgraphs such that any
edge of $G$ belongs to exactly
two of them.
\end{conjecture}

A $k$-edge-coloring of a graph $G$ is an assignment of colors $\{1,...,k\}$ to edges of $G$, such that adjacent edges receive different colors. If $c$ is an edge-coloring of $G$, then for a vertex $v$ of $G$, let $S_{c}(v)$ be the set of colors that edges incident to $v$ receive. 

\begin{definition}\label{def:poorrich}
Let $uv$ be an edge of a cubic graph $G$ and $c$ be an edge-coloring of $G$. Then the edge $uv$ is called {\bf poor} or {\bf rich} with respect to $c$, if $|S_{c}(u)\cup S_{c}(v)|=3$ or $|S_{c}(u)\cup S_{c}(v)|=5$, respectively. 
\end{definition}

Edge-colorings having only poor edges are trivially $3$-edge-colorings of $G$. Also edge-colorings having only rich edges have been considered in the last years, and they are called strong edge-colorings.
In this paper, we will focus on the case when all edges must be either poor or rich.

\begin{definition}\label{def:normal}
An edge-coloring $c$ of a cubic graph is {\bf normal}, if any edge is rich or poor with respect to $c$. 
\end{definition} 

It is straightforward that an edge coloring which assigns a different color to every edge of a simple cubic graph is normal since all edges are rich. Hence, we can define the normal chromatic index of a simple cubic graph $G$, denoted by $\chi'_{N}(G)$, as the smallest $k$, for which $G$ admits a normal $k$-edge-coloring. In \cite{Jaeger1985}, Jaeger has shown that:

\begin{proposition}\label{prop:JaegerNormalColor}(Jaeger, \cite{Jaeger1985}) If $G$ is a cubic graph, then $P_{10}\prec G$, if and only if $G$ admits a normal $5$-edge-coloring.
\end{proposition} This implies that Conjecture \ref{conj:P10conj} can be stated as follows:

\begin{conjecture}\label{conj:5NormalConj} For any bridgeless cubic graph $G$, $\chi'_{N}(G)\leq 5$.
\end{conjecture} Observe that Conjecture \ref{conj:5NormalConj} is trivial for $3$-edge-colorable cubic graphs. This is true because in any $3$-edge-coloring $c$ of a cubic graph $G$ any edge $e$ is poor, hence $c$ is a normal edge-coloring of $G$. Thus non-$3$-edge-colorable cubic graphs are the main obstacle for Conjecture \ref{conj:5NormalConj}. Note that Conjecture \ref{conj:5NormalConj} is verified for some non-$3$-edge-colorable bridgeless cubic graphs in \cite{HaggSteff2013}. Finally, let us note that in \cite{Samal2011} the percentage of edges of a bridgeless cubic graph, which can be made poor or rich in a 5-edge-coloring, is investigated.

If we consider the larger class of simple cubic graphs, without any assumption on connectivity, some interesting questions naturally arise. Indeed, examples of simple cubic graphs with $\chi'_{N}(G) > 5$ can be constructed in this class, and hence it is natural to ask for a possible upper bound for this parameter. 

Let us remark that any strong edge-coloring is, in particular, a normal edge-coloring. Andersen has shown in \cite{Andersen1992} that any simple cubic graph admits a strong edge-coloring with ten colors, hence ten is also an upper-bound for the normal chromatic index. The result was improved, following the approach of Andersen, in \cite{Bilkova12}, where it is shown that any simple cubic graph admits a normal edge-coloring with nine colors.
In this paper, we prove that if $G$ is any simple cubic graph, then $\chi'_{N}(G)\leq 7$. We complement this result by constructing an infinite family of simple cubic graphs with $\chi'_{N}(G)= 7$. Thus our result is best-possible. 

\section{Some Auxiliary Results}
\label{sec:aux}

In this section, we present some results that will be helpful in obtaining Theorem \ref{thm:EnblockImplies7SimpleCase} which is the main result of this paper.





\begin{theorem}\label{thm:Jaeger8flow} (Jaeger, \cite{Jaeger1975,Jaeger1979}) Any bridgeless graph admits a nowhere-zero $\mathbb{Z}^3_{2}$-flow.
\end{theorem}

%
%
%

We will also need to recall a classical theorem of Nash-Williams and Tutte about disjoint spanning trees.

\begin{theorem}
\label{thm:NashWilliams} (\cite{Zhang1997}) Let $G$ be a graph and $k\geq 1$. Then $G$ contains $k$ edge-disjoint spanning trees, if and only if for any partition $P=(V_1,...,V_t)$ of $V(G)$, $|E_c(P)| \geq k(t-1)$. Here $E_c(P)$ denotes the set of edges of $G$ that connect two vertices that lie in different $V_i$s.
\end{theorem}

Below we prove two lemmas about nowhere zero $\mathbb{Z}_2^2$-flows of arbitrary $4$-edge-connected graphs. See exercises 3.13 and 3.14 from \cite{Zhang1997} for similar statements.

From now on, we denote by $\{x,y\}$ a set of generators of the group $\mathbb{Z}_2^2$ , while we denote by $\{x,y,z\}$ a set of generators of the group $\mathbb{Z}_2^3$.

\begin{lemma}
\label{lem:TwoEdgeEqual} 
Let $G$ be a $4$-edge-connected (pseudo)graph, and let $e$ and $f$ be two edges of $G$. Then $G$ admits a nowhere zero $\mathbb{Z}_2^2$-flow $\theta$, such that $\theta(e)=\theta(f)$.  
\end{lemma}

\begin{proof} We will assume that $e$ and $f$ are not loops, otherwise the statement is trivial since the flow value of a loop can be arbitrarily chosen in $\{x,y,x+y\}$. Consider the graph $G-e-f$. Let us show that it has two edge-disjoint spanning trees. We will use Theorem \ref{thm:NashWilliams}. Consider any partition $P=(V_1,...,V_t)$ of $V(G)$. Let us count the number of edges crossing the sets $V_i$s, that is $|E_c(P)|$. Since $G$ is $4$-edge-connected, any fixed $V_i$ is connected with the rest of the graph $G$ with at least four edges. At most two of these edges can be $e$ and $f$, therefore
\[|E_c(P)|\geq \frac{4t}{2} - 2 =2t-2=2(t-1).\]
Thus by Theorem \ref{thm:NashWilliams}, $G-e-f$ has two edge-disjoint spanning trees, say $T_1$ and $T_2$. Clearly, $T_1$ and $T_2$ are also disjoint spanning trees of $G$. For every edge $g \in E(G)\setminus E(T_1)$, we denote by $C^1_g$ the unique cycle in $T_1+g$. Analogously, we denote by $C^2_h$ the unique cycle in $T_2+h$ for every $h \in E(G) \setminus E(T_2)$. We construct a nowhere-zero $\mathbb{Z}^2_2$-flow $\theta$ of $G$ by adding $x$ on each edge of the cycles $C^1_g$, where $g \notin E(T_1)$ and by adding $y$ on each edge of the cycles $C^2_h$, where $h \notin E(T_2)$. 
Hence, both edges $e$ and $f$ receive value $x+y$ in $\theta$, since both of them belong neither to $E(T_1)$ nor to $E(T_2)$.     
%
%
%
%
%
%
%
\end{proof}

\begin{lemma}
\label{lem:3edgesAdjacent4connected} Let $G$ be a $4$-edge-connected (pseudo)graph, and let $e,f,g$ be three
edges incident to some vertex $v$ of $G$. Then $G$ has a nowhere zero
$\mathbb{Z}_2^2$-flow $\theta$, such that $\theta(e)\neq \theta(f)$ and $\theta(e)\neq \theta(g)$.
\end{lemma}

\begin{proof} 
We construct a nowhere zero $\mathbb{Z}_2^2$-flow arising from two disjoint even subgraphs of $G$ in the standard way (see Theorem 3.2.4 in \cite{Zhang1997}). One can easily see that if one of the even subgraphs does not contain $e$ and does contain $f,g$, then the obtained flow meets our constraints. Now, we construct two even subgraphs $P_1$ and $P_2$ which satisfy such a condition.

Firstly, we assume that none of $e,f,g$ is a loop, as otherwise the statement of the lemma is trivial. 
From the proof of the previous lemma, we have that $G-e-f$ has two edge-disjoint spanning trees, say $T_1$ and $T_2$, and  without loss of generality we can assume $g \notin T_2$.

Since a spanning tree of a graph contains a parity subgraph of the graph (see Lemma 3.2.8 in \cite{Zhang1997}), we can choose two parity subgraphs of $G$, say $A_1$ and $A_2$, contained in $T_1$ and $T_2$, respectively. Let $C$ be the unique cycle in the subgraph $T_2 \cup \{e\}$. It is straightforward that $e \in C$. Denote by $P_1$ the even subgraph of $G$ which is the complement of $A_1$ and by $P_2$ the even subgraph of $G$ which is the complement of the parity subgraph $A_2 \bigtriangleup C$. Since $e \in C$ and $e \notin  A_2$, it follows that $e$ does not belong to $P_2$. On the other hand, $f,g$ do not belong to $T_2 \cup {e}$ hence they belong to $P_2$.
\end{proof}

\begin{corollary}
\label{cor:TwoNotEqualedges} Let $G$ be a $4$-edge-connected (pseudo)graph, and let $e$ and $f$ be two edges incident to the same vertex $v$. Then $G$ has a nowhere zero $\mathbb{Z}_2^2$-flow $\theta$, such that $\theta(e)\neq \theta(f)$.  
\end{corollary}

\section{The Main Result}
\label{sec:Main}

In this section we present our main result. Conjecture \ref{conj:5NormalConj} states that $\chi'_{N}(G)\leq 5$ for any bridgeless cubic graph. Combined with Proposition \ref{prop:JaegerNormalColor} and the fact that any cubic graph admitting a $P_{10}$-coloring, has to be bridgeless, we have that if $G$ is a cubic graph with a bridge, then $\chi'_{N}(G)\geq 6$. The following theorem presents a way to construct infinitely many cubic graphs containing bridges, such that $\chi'_{N}(G)\geq 7$.

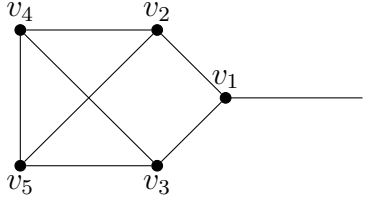
\begin{figure}[!htbp]
\begin{center}

\begin{tikzpicture}[scale=0.45]


 \node at (0, 0.55) {$v_1$};
   \node at (-2, 2.55) {$v_2$};
   \node at (-2, -2.55) {$v_3$};
   \node at (-6, 2.55) {$v_4$};
    \node at (-6, -2.55) {$v_5$};

 \tikzstyle{every node}=[circle, draw, fill=black!,
                        inner sep=0pt, minimum width=4pt]
    \node at (0,0) (n1) {};


   \node at (-2,2) (n2) {};

   
    \node at (-2,-2) (n3) {};

    
     \node at (-6,2) (n4) {};
    
   
     \node at (-6,-2) (n5) {};

 \draw (n1)--(n2);
 \draw (n1)--(n3);
 
 \draw (n2)--(n4);
 \draw (n3)--(n4);
 \draw (n3)--(n5);
 \draw (n2)--(n5);
 
 \draw (n4)--(n5);
 
 \draw (0,0) -- (4,0);

\end{tikzpicture}
\end{center}
\caption{A subgraph in a cubic graph that requires $7$ colors in a normal coloring.} \label{fig:exam}
\end{figure}

\begin{theorem}\label{thm:7examplethm} Let $K$ be the graph obtained from $K_4$ by subdividing one of its edges once (Figure \ref{fig:exam}). Then for any cubic graph $G$ containing $K$ as a subgraph, one has 
\begin{enumerate}[(a)]
    \item in any normal edge-coloring of $G$, the edges of $K$ are rich,

    \item in any normal edge-coloring of $G$, the edges of $K$ are colored with pairwise different colors,
    
    \item $\chi'_{N}(G)\geq 7$,
    
    \item in any normal $7$-edge-coloring of $G$, the colors of $v_4v_5$ and the bridge incident to $v_1$ have to be the same.
\end{enumerate}
\end{theorem}

\begin{proof} Observe that it suffices to prove only the statement (a). The other three statements follow easily by a direct check.

Let $c$ be any normal edge-coloring of $G$. Let us show that the edge $v_2v_5$ is rich. Assume that it is poor. Without loss of generality we can assume that $c(v_2v_5)=1$, $c(v_2v_4)=2$ and $c(v_1v_2)=3$. Since $c$ is an edge-coloring, we have $c(v_2v_4)\neq c(v_4v_5)$, hence $c(v_4v_5)=3$ and $c(v_3v_5)=2$. Since $c(v_3v_5)=c(v_2v_4)$, we have that the edge $v_3v_4$ is poor, too. Hence $3=c(v_4v_5)=c(v_1v_3)$, which is a contradiction that $c$ is an edge-coloring. 

By symmetry of $K$ the edges $v_2v_4$, $v_3v_4$ and $v_3v_5$ are also rich with respect to $c$.

Now, let us show that the edge $v_4v_5$ is also rich. Assume that it is poor. Without loss of generality, we can assume that $c(v_4v_5)=1$, $c(v_3v_4)=2$ and $c(v_2v_4)=3$. Since $c(v_2v_5)\neq c(v_2v_4)$, we have $c(v_2v_5)=2$ and $c(v_3v_5)=3$. Consider the edge $v_2v_4$. Observe that it is adjacent to two edges of color $2$, hence it should be poor, which is a contradiction.

Finally, let us show that the edge $v_1v_2$ has to be rich. Again assume that it is poor. Without loss of generality, we can assume that $c(v_1v_2)=1$ and $c(v_1v_3)=2$. Observe that one of edges $v_2v_4$ or $v_2v_5$ has to have color $2$. If $c(v_2v_4)=2$, then the edge $v_3v_4$ is poor, which is a contradiction. On the other hand, if $c(v_2v_5)=2$, then the edge $v_3v_5$ is poor, which is a contradiction.
Again by symmetry of $K$ we have that $v_1v_3$ is also rich with respect to $c$, and the assertion follows.
\end{proof}

We now proceed with showing that $\chi'_{N}(G)\leq 7$ for any simple cubic graph $G$. Observe that combined with the previous theorem, we will have that the upper bound seven is best-possible. First we recall a proof of this bound for bridgeless cubic graphs, which is an easy application of Jaeger's $8$-flow theorem (Theorem \ref{thm:Jaeger8flow}). Let us note that this proof has been already proposed in \cite{Bilkova12}. (See also Theorem 1.1 in \cite{HolySkoJCTB2004}). We start with the following easy remark:

\begin{remark}
\label{rem:Normal} Let $G$ be a cubic graph. If $c$ is an edge-coloring of $G$, such that $c(e_3)$ is uniquely determined by $c(e_1)$ and $c(e_2)$, then $c$ is a normal edge-coloring. Here $e_1, e_2, e_3$ are the three edges of $G$ incident to the same vertex $v$.
\end{remark}

\begin{theorem}\label{thm:7bridgeless} If $G$ is a bridgeless cubic graph, then $\chi'_{N}(G)\leq 7$. 
\end{theorem}

\begin{proof} By Theorem \ref{thm:Jaeger8flow}, $G$ admits a nowhere-zero $\mathbb{Z}^3_{2}$-flow $\phi$. Let $e_1, e_2, e_3$ be three edges of $G$ incident to the same vertex $v$. It is easy to see that the values of $\phi$ on any two of $e_1, e_2, e_3$ uniquely determine the the value of $\phi$ on the third one. Thus, $\phi$ is a normal $7$-edge-coloring thanks to Remark \ref{rem:Normal}.
\end{proof}

%
%

Observe that the proof of the previous theorem suggests that any nowhere zero $\mathbb{Z}^3_{2}$-flow of the bridgeless cubic graph $G$ gives rise to a normal $7$-edge-coloring of $G$. If an edge is rich or poor in this coloring, we will simply say that this edge is rich or poor, respectively, in the corresponding nowhere zero $\mathbb{Z}^3_{2}$-flow. Our next result states that one can make an arbitrary fixed edge of a bridgeless cubic graph poor in a nowhere zero $\mathbb{Z}_2^3$-flow.

In the proof, and in the rest of the paper, we will use several times the following standard operations on cubic graphs.

\begin{itemize}
\item Given two cubic graphs $G_1$ and $G_2$ and two edges $x_1y_1$ in $G_1$ and $x_2y_2$ in $G_2$, the {\bf 2-cut-connection} of $(G_1,x_1,y_1)$ and $(G_2,x_2,y_2)$ is the graph obtained from $G_1$ and $G_2$ by removing edges $x_1y_1$ and $x_2y_2$, and connecting $x_1$ and $y_1$ by a new edge, and $x_2$ and $y_2$ by another new edge.
On the other hand, if a cubic graph $G$ has a $2$-edge-cut $C$, we refer to $G_1$ and $G_2$ as the graphs obtained from $G$ by the {\bf 2-cut reduction of $C$}.

\item Given two cubic graphs $G_1$ and $G_2$ and two vertices $u_1$ of $G_1$ and $u_2$ of $G_2$, a {\bf star product} of $(G_1,u_1)$ and $(G_2,u_2)$ is a cubic 
graph obtained from $G_1$ and $G_2$ by removing vertices $u_1$ and $u_2$, and
connecting the three neighbors of $u_1$ in $G_1$ to the three neighbors of $u_2$ in $G_2$ with three new independent edges.
On the other hand, if a cubic graph $G$ has a non-trivial $3$-edge-cut $C$ we refer to $G_1$ and $G_2$ as the graphs obtained from $G$ by a {\bf 3-cut reduction of $C$}.
\end{itemize}

\begin{remark}\label{rem:cutreduction}
In what follows, with a slightly abuse of terminology, we will always consider an edge of $G$ not in $C$ also as an edge of either $E(G_1)$ or $E(G_2)$. While we will refer to the other edges of $G_1$ and $G_2$ as the edges which arise from $C$. 
\end{remark}

Moreover, the following refinement of Petersen Theorem for perfect matchings in cubic graphs will be used in the proof of next two lemmas.

\begin{theorem}\label{PetersenTheorem}(\cite{Sch}) Any edge of a bridgeless cubic graph $G$ lies in a perfect matching of $G$.
\end{theorem}

Finally, we will also make use several times of some properties of the automorphism group of the elementary abelian group $\mathbb{Z}_2^3$. In particular, we need to use the following standard remark.
\begin{remark}\label{rem:automorphisms}
If $S_1$ and $S_2$ are sets of generators of $\mathbb{Z}_2^3$ of cardinality three, then any bijective map from $S_1$ to $S_2$ can be uniquely extended to an automorphism of $\mathbb{Z}_2^3$.   
\end{remark}

\begin{lemma}
\label{lem:1edgePoor} Let $G$ be a bridgeless cubic graph, and $e$ be a prescribed edge. Then there is a nowhere zero $\mathbb{Z}_2^3$-flow $\theta$, such that $e$ is poor in $\theta$.
\end{lemma}

\begin{proof} Consider a possible counterexample $G$ with the minimum number of vertices. Clearly, $G$ is connected. Let us show that it has no $2$-edge-cuts. By contradiction, assume $C$ is a $2$-edge-cut of $G$. Consider the cubic graphs $G_1$ and $G_2$ obtained by the $2$-cut reduction of $C$. Since $G_1$ and $G_2$ are  smaller than $G$, we have that they are not counterexamples. 

If $e\notin C$, we can assume that $e\in E(G_1)$ (see Remark \ref{rem:cutreduction}). Take a nowhere zero $\mathbb{Z}_2^3$-flow $\theta$, where $e$ is poor in $G_1$, and any nowhere zero $\mathbb{Z}_2^3$-flow $\mu$ of $G_2$. By choosing a suitable automorphism of $\mathbb{Z}_2^3$ (Remark \ref{rem:automorphisms}), we can assume that $\theta$ and $\mu$ agree on edges arising from $C$. Thus, we can easily construct a nowhere zero $\mathbb{Z}_2^3$-flow of $G$, where $e$ is poor. 

On the other hand, if $e\in C$, then assume $e=uv$ and let $e'=u'v'$ be the other edge of $C$. We assume that $u$ and $u'$ belong to the same component of $G-C$. A similar statement holds for $v$ and $v'$. Consider the cubic graphs $G_1$ and $G_2$ obtained by the $2$-cut reduction of $C$ by adding possibly parallel edges $e_1=uu'$ and $e_2=vv'$. Since $G_1$ and $G_2$ are smaller than $G$, we can make $e_1$ poor in a nowhere zero $\mathbb{Z}_2^3$-flow of $G_1$, and $e_2$ poor in a nowhere zero $\mathbb{Z}_2^3$-flow of $G_2$. By choosing a suitable automorphism of $\mathbb{Z}_2^3$ (Remark \ref{rem:automorphisms}), we can assume that these two flows have the same value on $e_1$ and $e_2$. Moreover, the values of these flows are the same on edges incident to $u$ and $v$ (hence on edges incident to $u'$ and $v'$). Now, we can easily construct a nowhere zero $\mathbb{Z}_2^3$-flow of $G$, where $e$ is poor. 

Thus, our counterexample is $3$-connected. Let us show that all $3$-edge-cuts in $G$ are trivial. Assume that there is a non-trivial $3$-edge cut $C$. Let us show that $e\in C$. On the opposite assumption, consider the two $3$-connected cubic graphs $G_1$ and $G_2$ obtained by a $3$-cut reduction of $C$. Assume that $e\in E(G_1)$. Since $G_1$ is not a counterexample, we have that $e$ can be made poor in a nowhere zero $\mathbb{Z}_2^3$-flow $\theta$ of $G_1$. Take an arbitrary nowhere zero $\mathbb{Z}_2^3$-flow of $G_2$. By choosing a suitable automorphism of $\mathbb{Z}_2^3$ (Remark \ref{rem:automorphisms}), we can have that these two flows agree on edges of $C$. But then, we will get a nowhere zero $\mathbb{Z}_2^3$-flow of $G$, where $e$ is poor contradicting our assumption that $G$ is a counterexample.

Thus, we can assume that $e\in C$. Again, consider the two $3$-connected cubic graphs $G_1$ and $G_2$ obtained by a $3$-cut reduction of $C$. Since $G_1$ and $G_2$ are smaller than $G$, we have that they are not counterexamples, hence $e$ can be made poor in a nowhere zero $\mathbb{Z}_2^3$-flow $\theta_i$ of $G_i$, $i=1,2$. By choosing a suitable automorphism of $\mathbb{Z}_2^3$ (Remark \ref{rem:automorphisms}), we can assume that $\theta_1$ and $\theta_2$ agree on edges of $C$. Now consider the nowhere zero $\mathbb{Z}_2^3$-flow $\phi$ arising from $\theta_1$ and $\theta_2$. Since $e$ is poor in both $\theta_i$, $\theta_1$ and $\theta_2$ agree on edges of $C$, we have that $e$ is poor in $\phi$. This contradicts our assumption that $G$ is a counterexample. 

Thus, we can assume that $G$ is cyclically $4$-edge-connected. Let $g$ be an edge adjacent to $e$. Consider a perfect matching $M$ containing $g$ (Theorem \ref{PetersenTheorem}). Observe that $\overline{M}$, the $2$-factor complementary to $M$, contains the edge $e$. Consider the pseudo-graph $H=G/E(\overline{M})$ obtained from $G$ by contracting the edges of $\overline{M}$. We keep the parallel edges and loops arising as a result of this. Since $G$ is cyclically $4$-edge-connected, we have that $H$ is $4$-edge-connected. Let $g_e$ be the edge of $M$ that is adjacent to $e$, and is different from $g$. By Lemma \ref{lem:TwoEdgeEqual}, $H$ admits a nowhere zero $\mathbb{Z}_2^2$-flow $\theta$, such that $\theta(g)= \theta(g_e)$.

We now extend $\theta$ to a nowhere zero $\mathbb{Z}_2^3$-flow $\mu$ of $G$ as follows (see the proof of Lemma 5.2 in \cite{HolySkoJCTB2004}): first for any edge $h\in M$, we define the triple $\mu(h)$ as follows: $\mu(h)=(0,\theta(h))$. Now, let $C$ be any cycle of $\overline{M}$. Let $x_0$ be any element of $\mathbb{Z}_2^3$, whose first coordinate is $1$. Assign $x_0$ to an edge of $C$. Then observe that the rest of the values of edges of $C$ are defined uniquely in $\mu$. Moreover, the first coordinate of the values of $\mu$ on $C$ is $1$. Hence for any edges $h_1\in M$ and $h_2\in \overline{M}$, we have $\mu(h_1)\neq \mu(h_2)$. Also observe that for different cycles of $\overline{M}$ we can choose $x_0$ differently.

Now, let us show that $\mu$ meets our constraints. Since $\theta(g)= \theta(g_e)$, we have $\mu(g)= \mu(g_e)$. Hence the two edges of $\overline{M}$ adjacent to $e$ must have the same value in $\mu$. Hence the edge $e$ is poor in $\mu$.
\end{proof}

Our next statement shows that any two adjacent edges of a $3$-connected cubic graph can be made rich in a nowhere zero $\mathbb{Z}_2^3$-flow. Note that the statement cannot be proved for all bridgeless cubic graphs (see example in Figure \ref{fig:ExampleNonrichedgeSimple6}).

\begin{lemma}
\label{lem:chi7twoedgeNowhereZero8flow} Let $G$ be a $3$-connected cubic graph, and let $e$ and $f$ be two adjacent edges of $G$. Then, $G$ admits a nowhere zero $\mathbb{Z}_2^3$-flow such that $e$ and $f$ are rich.
\end{lemma}

\begin{proof} Consider a possible counterexample $G$ with the minimum number of vertices. Since $G$ is $3$-connected, we have that any non-trivial $3$-edge cut should be a matching. Let us show that there are no non-trivial $3$-edge cuts in $G$. 

Assume $C$ is a non-trivial $3$-edge cut. Let us show that $C\cap \{e,f\}\neq \emptyset$. On the opposite assumption, consider the two $3$-connected cubic graphs $G_1$ and $G_2$ obtained by a $3$-cut reduction of $C$. Assume that $e,f\in E(G_1)$. Since $G_1$ is not a counterexample, we have that $e$ and $f$ can be made rich in a nowhere zero $\mathbb{Z}_2^3$-flow $\theta$ of $G_1$. Take arbitrary nowhere zero $\mathbb{Z}_2^3$-flow of $G_2$. By choosing a suitable automorphism of $\mathbb{Z}_2^3$ (Remark \ref{rem:automorphisms}), we can have that these two flows agree on edges of $C$. But then, we will get a nowhere zero $\mathbb{Z}_2^3$-flow of $G$, where $e$ and $f$ are rich contradicting our assumption that $G$ is a counterexample.

Thus, we can assume that $C\cap \{e,f\}\neq \emptyset$. Since $C$ is a matching, and $e$ and $f$ are adjacent to the same vertex, we have that only one on them belongs to $C$. Assume that it is $e$. Again, consider the two $3$-connected cubic graphs $G_1$ and $G_2$ obtained by a $3$-cut reduction of $C$. Assume that $f\in E(G_1)$. Since $G_1$ is smaller than $G$, $G_1$ is not a counterexample, hence $e$ and $f$ can be made rich in a nowhere zero $\mathbb{Z}_2^3$-flow $\theta$ of $G_1$. By Lemma \ref{lem:1edgePoor}, we can make $e$ poor in a nowhere zero $\mathbb{Z}_2^3$-flow $\mu$ of $G_2$. By choosing a suitable automorphism of $\mathbb{Z}_2^3$ (Remark \ref{rem:automorphisms}), we can assume that $\theta$ and $\mu$ agree on edges of $C$. Now consider the nowhere zero $\mathbb{Z}_2^3$-flow arising from $\theta$ and $\mu$. Since $e$ and $f$ were rich in $\theta$, $\theta$ and $\mu$ agree on edges of $C$, and $e$ was poor in $\mu$, we have that $e$ and $f$ are rich in $G$. This contradicts our assumption that $G$ is a counterexample. 

Thus, we can assume that all $3$-edge-cuts of $G$ are trivial. Hence $G$ is cyclically $4$-edge-connected. Let $g$ be the third edge adjacent to $e$ and $f$. Consider a perfect matching $M$ containing $g$ (Theorem \ref{PetersenTheorem}). Observe that $\overline{M}$, the $2$-factor complementary to $M$, contains the edges $e$ and $f$. Moreover, they lie in the same cycle of the $2$-factor. Consider the pseudo-graph $H=G/E(\overline{M})$ obtained from $G$ by contracting all edges of $\overline{M}$. We keep the parallel edges and loops arising as a result of this. Since $G$ is cyclically $4$-edge-connected, we have that $H$ is $4$-edge-connected. Let $g_e$ and $g_f$ be the edges of $M$ that are adjacent to $e$ and $f$, respectively, and are different from $g$. By Lemma \ref{lem:3edgesAdjacent4connected}, $H$ admits a nowhere zero $\mathbb{Z}_2^2$-flow $\theta$, such that $\theta(g)\neq \theta(g_e)$ and $\theta(g)\neq \theta(g_f)$.

We now extend $\theta$ to a nowhere zero $\mathbb{Z}_2^3$-flow $\mu$ of $G$ exactly in the way we did in the proof of Lemma \ref{lem:1edgePoor}. We have $\mu(g_e) \neq \mu (g)$. Thus, the edge $e$ is rich in $\mu$. Similarly, since $\theta(g)\neq \theta(g_f)$ by our choice, one can easily show that $f$ is rich in $\mu$.
\end{proof}



\begin{corollary}
\label{cor:chi7oneedgeNowhereZero8flow} Let $G$ be a $3$-connected cubic graph, and let $e$ be an edge. Then $G$ admits a nowhere zero $\mathbb{Z}_2^3$-flow $\theta$, such that $e$ is rich in $\theta$.
\end{corollary}

Now, we are going to consider simple graphs which are obtained from any bridgeless cubic graph by subdividing one of its edges and attaching a bridge to the new degree two vertex. The other end-vertex of the bridge has degree one. We are going to show that any such graph admits a normal edge-coloring with at most $7$ colors. Here the normality is understood in the following way: in the coloring adjacent edges receive different colors, all edges of the graph except the unique bridge must be poor or rich. However we do not impose any constraint on the bridge.

\begin{theorem}
\label{thm:EndblocksNormal7colors} Let $G'$ be a simple graph obtained from a bridgeless cubic graph $G$ by subdividing one of its edges once, adding a new vertex and adding an edge connecting the degree-two vertex with the new vertex. Then $\chi'_N(G')\leq 7$. 
\end{theorem}

\begin{proof} Let $G$ be a bridgeless cubic graph, and let $e=uw$ be any edge of $G$. We can assume that $G$ is connected. Consider the graph $G'$ obtained from $G$ by subdividing $e$ with a vertex $v_e$. The vertex $v_e$ is incident to the unique bridge in $G'$. We have that all degrees in $G'$ are three except the new vertex adjacent to $v_e$ which has degree one. Moreover, assume that $w_1$ and $w_2$ are the other two neighbors of $w$ in $G$ that differ from $u$.

First, we consider the case when $G$ is $3$-edge-connected. By Lemma \ref{lem:chi7twoedgeNowhereZero8flow}, there is a nowhere zero $\mathbb{Z}_2^3$-flow $\theta$, such that $ww_1$ and $ww_2$ are rich. Observe that since $\theta(ww_1)\neq \theta(ww_2)$, the two values of $\theta$ on edges incident to $w_1$ that differ from $ww_1$ cannot coincide with the two values of $\theta$ on edges that are incident to $w_2$ and differ from $ww_2$. Let us show that the intersection of these two sets is exactly one. We need to rule out the case when they are disjoint.

Assume that $w_1$ is incident to edges with flow values $x$ and $y$, and let $\theta(ww_1)=x+y$. Observe that $x+y$ cannot appear around $w_2$, as $\theta(ww_1)\neq \theta(ww_2)$ and $ww_2$ is rich. Let $z$ be an element of $\mathbb{Z}^3_2$ such that $z \notin \{0,x,y,x+y\}$. Then, $\mathbb{Z}^3_2=\{0,x,y,x+y\} \cup \{z,x+z,y+z,x+y+z\}$. The edges incident to $w_2$ that differ from $ww_2$ have flow value in $\{z,x+z,y+z,x+y+z\}$. Then, in any case, the flow value of $ww_2$ belongs to $\{x,y,x+y\}$, which is a contradicion since either we have two incident edges with the same flow value or $ww_1$ is not rich.

Thus, without loss of generality, we can assume that $w_1$ is incident to edges with flow values $x,y$ and $\theta(ww_1)=x+y$, $w_2$ is incident to edges with flow values $x,z$ and $\theta(ww_2)=x+z$, and $y\neq z$. By considering $\partial(\{w,w_1,w_2\})$, we have that $\theta(e)=y+z$. Let $t_1$ and $t_2$ be the two values of $\theta$ on edges incident to $u$ that differ from $e$. Clearly, $t_1+t_2=y+z$. Now, we are going to obtain a normal $7$-edge-coloring of $G'$ using the seven non-zero elements of $\mathbb{Z}_2^3$. We will consider two cases.

\medskip

Case 1: $\{t_1, t_2\}\cap \{x,y,z\}=\emptyset$. Let us show that we can assume that $\{t_1, t_2\}\cap \{x,y,z,x+z,x+y\}=\emptyset$. If not, we have that  $\{t_1, t_2\}= \{x+y,x+z\}$ and edge $e$ is poor in $\theta$. Extend $\theta$ to a normal $7$-edge-coloring $c$ of $G$ as follows: take $c$ equal to $\theta$ everywhere in $G'$, except $c(uv_e)=y+z$, $c(v_ew)=x+y+z$ and the value of $c$ on the unique bridge of $G'$ is $x$. It can be easily seen that $c$ is a normal $7$-edge-coloring of $G'$. 

Thus we can assume that $\{t_1, t_2\}\cap \{x,y,z,x+z,x+y\}=\emptyset$, that is $\{t_1, t_2\}=\{y+z,x+y+z\}$. Again, we have a contradiction since $\theta(e)=y+z$ and then we have two edges incident $u$ with value $y+z$. 

Case 2: $\{t_1, t_2\}\cap \{x,y,z\}\neq \emptyset$, that is either $\{t_1,t_2\}=\{x,x+y+z\}$ or $\{t_1,t_2\}=\{y,z\}$.
Extend $\theta$ to a normal $7$-edge-coloring $c$ of $G$ as follows: take $c$ equal to $\theta$ everywhere in $G'$, except $c(uv_e)=y+z$, $c(v_ew)=x+y+z$ and the value of $c$ on the unique bridge of $G'$ is $x$. It can be easily seen that $c$ is a normal $7$-edge-coloring of $G'$ in both cases: more precisely, if $\{t_1,t_2\}=\{x,x+y+z\}$ then $uv_e$ is poor and if $\{t_1,t_2\}=\{y,z\}$ then $uv_e$ is rich. 

Thus, it remains to consider the case when $G$ has a $2$-edge-cut. Let us prove the statement by induction on the number of vertices. 

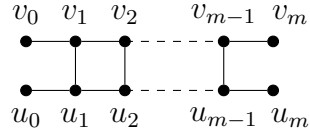
\begin{figure}[!htbp]
\begin{center}

\begin{tikzpicture}[scale=0.65]



\node at (0, -0.55) {$u_0$};
  \node at (0,1.55) {$v_0$};
  
  \node at (1, -0.55) {$u_1$};
  \node at (1, 1.55) {$v_1$};
  
  \node at (2, -0.55) {$u_2$};
  \node at (2,1.55) {$v_2$};

\node at (4, -0.55) {$u_{m-1}$};
  \node at (4,1.55) {$v_{m-1}$};

 \node at (5.35, -0.55) {$u_m$};
  \node at (5.35,1.55) {$v_m$};

 \tikzstyle{every node}=[circle, draw, fill=black!,
                        inner sep=0pt, minimum width=4pt]
    \node at (0,0) (n00) {};
    \node at (0,1) (n01) {};
    
    \node at (1,0) (n10) {};
    \node at (1,1) (n11) {};
    
    \node at (2,0) (n20) {};
    \node at (2,1) (n21) {};
    
    \node at (4,0) (n40) {};
    \node at (4,1) (n41) {};
    
    \node at (5,0) (n50) {};
    \node at (5,1) (n51) {};


 \draw (n00)--(n10);
 \draw (n11)--(n01);
 \draw (n11)--(n10);
 
 \draw (n11)--(n21);
 \draw (n10)--(n20);
 \draw (n20)--(n21);
 
 \draw (n40)--(n41);
 \draw (n40)--(n50);
 \draw (n41)--(n51);

 \draw[dashed] (2,0) -- (4,0);
 \draw[dashed] (2,1) -- (4,1);

\end{tikzpicture}
\end{center}
\caption{The $m$-ladder $L$ with initial vertices $u_0, v_0$ and terminal vertices $u_m, v_m$.} \label{fig:Ladder}
\end{figure}

For a positive integer $m$ define an $m$-ladder of $G$ as a subgraph $L$ (Figure \ref{fig:Ladder}) of $G$, such that: 
\begin{itemize}
 \item $V(L)=\{u_0,v_0,...,u_m, v_m\},$
 \item $E(L)=\{u_0u_1,u_1u_2,...,u_{m-1}u_m,v_0v_1,v_1v_2,...,v_{m-1}v_m,u_1v_1,...,u_{m-1}v_{m-1}\},$
 \item $u_0v_0\notin E(G), u_mv_m\notin E(G),$
 \item the set $\{u_iu_{i+1}, v_iv_{i+1}\}$ is a $2$-edge-cut of $G$ for all $i \in \{0,1,\ldots,m-1\}$,
 \item $u_0$, $v_0$ belong to the same component of $G-u_0u_1-v_0v_1$,
 \item $u_m$, $v_m$ belong to the same component of $G-u_{m-1}u_{m}-v_{m-1}v_m$.
\end{itemize}

Observe that since $G$ is bridgeless and cubic, for each 2-edge-cut $C$ in $G$ there is a positive integer $m$ and an $m$-ladder $L$ of $G$, such that $C\subseteq E(L)$. Moreover, $L$ is an induced subgraph of $G$. Indeed, following the notation for vertices and edges of $L$ introduced above, for each $i \in \{0,1,\ldots,m-1\}$ the pair of edges $u_iu_{i+1}$ and $v_iv_{i+1}$ is a $2$-edge-cut of $G$ (and $C$ is one of such pairs) and we have that $u_0$, $v_0$ belong to the same component of $G-u_{i}u_{i+1}-v_{i}v_{i+1}$ for any $i$ (and the same holds for $u_m$, $v_m$).
%
%


Assume that there exists a ladder $L$ such that the edge $e$ does not belong to $E(L)$.
Denote by $u_0$ and $v_0$ the initial vertices of $L$, which belong to a component $G_1$ of $G-E(L)$. Similarly, denote by $u_m$ and $v_m$ the terminal vertices of $L$, which belong to a component $G_2$ of $G-E(L)$.
Since $e$ does not lie in $L$, then it must lie either in $G_1$ or $G_2$. For the sake of definiteness, let $e\in E(G_1)$. Consider the cubic graph $H$ obtained from $G_1$ by adding the edge $u_0v_0$. By the definition of $L$, we have that $|V(H)|<|V(G)|$ and the graph $H'$ obtained from $H$ by subdividing the edge $e$ and attaching a pendant edge to $v_e$ is simple. Thus, by induction hypothesis, $\chi'_N(H')\leq 7$. Let $x$ be a non-zero element of $\mathbb{Z}_2^3$ such that the color of $u_0v_0$ in $H'$ is $x$. Now, consider a graph $H_1$ obtained from $G$ by removing the vertices of $G_1$ and adding a possibly parallel edge $u_1v_1$. Observe that $H_1$ is a bridgeless cubic graph, hence by Theorem \ref{thm:7bridgeless} it has a normal $7$-coloring arising from a nowhere zero $\mathbb{Z}_2^3$-flow of $H_1$. By renaming the colors in $H'$, we can assume that the color of $u_1v_1$ is $x$, and that the two colors incident to $u_1$ in $H_1$ coincide with two other colors incident to $u_0$ in $H'$. Now, consider an edge-coloring of $G'$ obtained from normal edge-colorings of $H'$ and $H_1$ by coloring the edges $u_0u_1$ and $v_0v_1$ with $x$. Observe that $u_0u_1$ is poor in $G'$, moreover, if $u_0v_0$ was poor in $H'$ or $u_1v_1$ was poor in $H_1$, then the new coloring is a normal $7$-edge-coloring of $G'$. On the other hand, if both $u_0v_0$ and $u_1v_1$ were rich in $H'$ and $H_1$, respectively, then we can always rename the colors in $H'$, so that the colors incident at $v_0$ in $H'$ coincide with the colors incident at $v_1$ in $H_1$. In the latter case, we will have that the edge $v_0v_1$ is poor.

Thus, we can assume that in $G$ for any ladder $L$ we have $e\in E(L)$. Consider a 2-edge-cut in $G$ and a ladder $L$ containing it. Define $G_1$, $G_2$ as the components of $G-E(L)$ which contain $u_0, v_0$ and $u_m, v_m$, respectively. Observe that the graphs $G_1+u_0v_0$ and $G_2+u_mv_m$ are simple. Let us show that they are $3$-edge-connected.
We prove this only for $G_1+u_0v_0$. Observe that $G_1+u_0v_0$ is bridgeless. Let us show that it has no a $2$-edge-cut. On the opposite assumption, consider a $2$-edge-cut $C_1$ of $G_1+u_0v_0$. If $u_0v_0\notin C_1$, then consider the ladder $L_1$ of $G$ containing the edges of $C_1$. Observe that $C_1$ is a $2$-edge-cut of $G$, such that the ladder $L_1$ containing it does not contain the edge $e$. This is a contradiction that $e$ must lie in all such ladders of $G$. Thus, we can assume that $u_0v_0\in C_1$. In this case the sets $\overline{C_1}=(C-u_0v_0)+u_0u_1$ and $\overline{C_2}=(C-u_0v_0)+v_0v_1$ are $2$-edge-cuts of $G$. Let $\overline{L_1}$ and $\overline{L_2}$ be the ladders of $G$ containing $\overline{C_1}$ and $\overline{C_2}$, respectively. Observe that at least one of them does not contain the edge $e$, which again contradicts our assumption. Thus, the graphs $G_1+u_0v_0$ and $G_2+u_mv_m$ are $3$-edge-connected.

Now, we are going to show a normal $7$-edge-coloring of $G'$. The edges $u_0u_1$, $v_0v_1$, $u_{m-1}u_m$ and $v_{m-1}v_m$ are called initial edges of $L$. The other edges of $L$ are called internal edges. First let us show the coloring of $G'$, when $e$ is an initial edge. Observe that this case includes the case when $L$ is comprised of two disjoint edges forming a 2-edge-cut. Assume that $e=u_{m-1}u_m$, where $u_{m-1}=u$ and $u_m=w$. Let us consider two graphs $H_1$ and $H_2$ obtained as follows: $H_1$ is obtained from the component of $G-u_{m-1}u_m-v_{m-1}v_m$ containing the vertex $u$ by adding a possibly parallel edge $u_{m-1}v_{m-1}$, and $H_2$ is the $3$-edge-connected graph $G_2+u_mv_m$. Observe that $H_1$ is a bridgeless cubic graph. Let $\theta_1$ be any nowhere zero $\mathbb{Z}_2^3$-flow of $H_1$, and let $\theta_2$ be a nowhere zero $\mathbb{Z}_2^3$-flow of $H_2$, such that the edges $ww_1$ and $ww_2$ are rich in $\theta_2$ (Lemma \ref{lem:chi7twoedgeNowhereZero8flow}). Here $w_1$ and $w_2$ are the neighbors of $w$ in $H_2$ that differ from $v_m$. By choosing a suitable automorphism of $\mathbb{Z}_2^3$ (Remark \ref{rem:automorphisms}), we can assume that $\theta_1(u_{m-1}v_{m-1})=\theta_2(u_mv_m)=x$. Moreover, the values of $\theta_1$ on the other two edges incident to $v_{m-1}$ agree with the values of $\theta_2$ on the other two edges incident to $v_m$. Now, we color the edge $v_{m-1}v_m$ and $uv_e$ with $x$, and extend it to a normal $7$-edge-coloring of $G'$ by considering the same strategy as we had in the case of $3$-connected graphs. 

Thus, it remains to consider the case when $e$ is an internal edge of $L$. The internal edges of $L$ are of two types, which we will naturally call horizontal and vertical edges (see Figure \ref{fig:Ladder}). For each of these cases we will exhibit a normal $7$-edge-coloring.

First, let us consider the case when the edge $e$ is a horizontal edge of $L$. We can assume that $e=u_{i-1}u_i$. As above, let $G_1$ be the component of $G-E(L)$ containing $u_0$ and $v_0$, and similarly, let $G_2$ be the component of $G-E(L)$ containing $u_m$ and $v_m$. We have that the cubic graphs $G_1+u_0v_0$ and $G_2+u_mv_m$ are $3$-edge-connected. Let $P_1$ be the shortest path of $L$ connecting $u_0$ and $u_{i-1}$, and let $P_2$ be the shortest path of $L$ connecting $u_i$ and $u_m$. Define the vertices $w'\in \{u_0, v_0\}$ and $w''\in \{u_m, v_m\}$ as follows: if the length of $P_1$ is odd, then $w'=v_0$, otherwise, $w'=u_0$, similarly, if the length of $P_2$ is odd, then $w''=v_m$, otherwise, $w''=u_m$. Since the cubic graphs $G_1+u_0v_0$ and $G_2+u_mv_m$ are $3$-edge-connected, Lemma \ref{lem:chi7twoedgeNowhereZero8flow} implies that these graphs have nowhere zero $\mathbb{Z}_2^3$-flows $\theta_1$ and $\theta_2$, such that the two edges incident to $w'$ and the two edges incident to $w''$ that differ from $u_0v_0$ and $u_mv_m$, respectively, are rich. By choosing a suitable automorphism of $\mathbb{Z}_2^3$ (Remark \ref{rem:automorphisms}), we can assume that $\theta_1(u_{0}v_{0})=\theta_2(u_mv_m)=x$. Consider the four edges of $G_1$ that are adjacent to an edge that is incident to $w'$. As we have shown in the analysis of the 3-edge-connected case (third, forth paragraphs), $\theta_1$ cannot have seven different values on these four edges together with two edges incident to $w'$ and the edge $u_0v_0$. Thus, there is a non-zero element $y$ of $\mathbb{Z}_2^3$, that does not appear on these seven edges. Similarly, define the element $z$ of $\mathbb{Z}_2^3$ as a value such that $\theta_2$ does not attain it on four edges of $G_2$ that are adjacent to an edge that is incident to $w''$, the two edges incident to $w''$ and the edge $u_mv_m$. Observe that $y \neq x$ and $z \neq x$, and by choosing a suitable automorphism (Remark \ref{rem:automorphisms}), we can not only assume $\theta_1(u_{0}v_{0})=\theta_2(u_mv_m)$, but also $y \neq z$. 

\begin{figure}[!htbp]
\begin{center}

\begin{tikzpicture}



\node at (-0.5, 0.5) {$G_1$};
\node at (13.5, 0.5) {$G_2$};

\node at (5.75, 1.25) {$y$};
\node at (7.25, 1.25) {$z$};
\node at (6.5, -0.25) {$x$};
\node at (4.5, -0.25) {$y$};
\node at (4.5, 1.25) {$x$};
\node at (8.5, -0.25) {$z$};
\node at (8.5, 1.25) {$x$};
\node at (3.5, -0.25) {$x$};
\node at (3.5, 1.25) {$y$};
\node at (9.5, 1.25) {$z$};
\node at (9.5, -0.25) {$x$};

\node at (3.5, 0.5) {$x+y$};
\node at (5.5, 0.5) {$x+y$};
\node at (7.5, 0.5) {$y+z$};
\node at (9.5, 0.5) {$y+z$};

\node at (7, 1.75) {$y+z$};

 \tikzstyle{every node}=[circle, draw, fill=black!,
                        inner sep=0pt, minimum width=4pt]
    \node at (0,0) (n00) {};
    \node at (1,0) (n10) {};
    \node at (4,0) (n40) {};
    \node at (5,0) (n50) {};
    \node at (8,0) (n80) {};
    \node at (9,0) (n90) {};
    \node at (12,0) (n120) {};
    \node at (13,0) (n130) {};
    
    \node at (0,1) (n01) {};
    \node at (1,1) (n11) {};
    \node at (4,1) (n41) {};
    \node at (5,1) (n51) {};
    \node at (6.5,1) (n651) {};
    \node at (8,1) (n81) {};
    \node at (9,1) (n91) {};
    \node at (12,1) (n121) {};
    \node at (13,1) (n131) {};


 \draw (n00)--(n10);
 \draw[dashed] (n10)--(n40);
 \draw (n40)--(n50);
 \draw (n50)--(n80);
 \draw (n80)--(n90);
 \draw[dashed] (n90)--(n120);
 \draw (n120)--(n130);
 
 \draw (n01)--(n11);
 \draw[dashed] (n11)--(n41);
 \draw (n41)--(n51);
 \draw (n51)--(n651);
 \draw (n651)--(n81);
 \draw (n81)--(n91);
 \draw[dashed] (n91)--(n121);
 \draw (n121)--(n131);
 
 \draw (n10)--(n11);
 \draw (n40)--(n41);
 \draw (n50)--(n51);
 \draw (n80)--(n81);
 \draw (n90)--(n91);
 \draw (n120)--(n121);

 \draw (6.5,1) -- (6.5,2);
 \draw (-0.5,0.5) circle (1cm);
 \draw (13.5,0.5) circle (1cm);
 

\end{tikzpicture}
\end{center}
\caption{The normal $7$-edge-coloring in the horizontal case.} \label{fig:A1}
\end{figure}
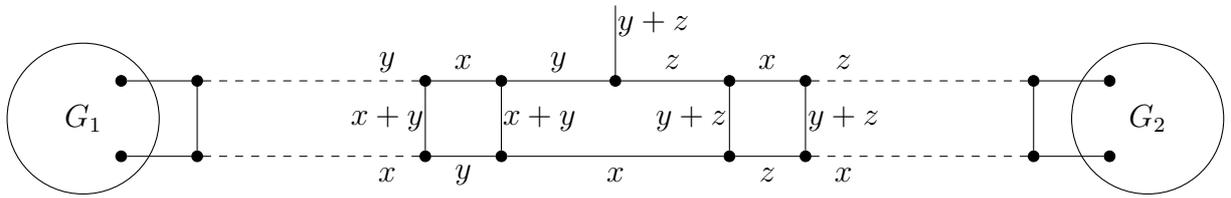

We extend the flows $\theta_1$ and $\theta_2$ to a normal $7$-edge-coloring of $G'$ as it is shown on Figure \ref{fig:A1}.
Moreover, the $u_0-u_{i-1}$ and $v_0-v_{i-1}$ subpaths of $L$ are colored $x-y$, alternatively. Similarly, the $u_i-u_{m}$ and $v_i-v_{m}$ subpaths of $L$ are colored $x-z$, alternatively.

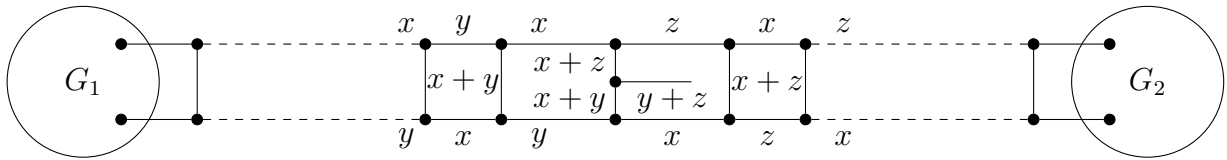
\begin{figure}[!htbp]
\begin{center}

\begin{tikzpicture}



\node at (-0.5, 0.5) {$G_1$};
\node at (13.5, 0.5) {$G_2$};

\node at (4.5, 0.5) {$x+y$};
\node at (8.5, 0.5) {$x+z$};

\node at (5.9, 0.25) {$x+y$};
\node at (5.9, 0.75) {$x+z$};
\node at (7.25, 0.25) {$y+z$};

\node at (3.75, -0.25) {$y$};
\node at (4.5, -0.25) {$x$};
\node at (5.5, -0.25) {$y$};
\node at (7.25, -0.25) {$x$};
\node at (8.5, -0.25) {$z$};
\node at (9.5, -0.25) {$x$};  

\node at (3.75, 1.25) {$x$};
\node at (4.5, 1.25) {$y$};
\node at (5.5, 1.25) {$x$};
\node at (7.25, 1.25) {$z$};
\node at (8.5, 1.25) {$x$};
\node at (9.5, 1.25) {$z$}; 
 
 \tikzstyle{every node}=[circle, draw, fill=black!,
                        inner sep=0pt, minimum width=4pt]
    \node at (0,0) (n00) {};
    \node at (1,0) (n10) {};
    \node at (4,0) (n40) {};
    \node at (5,0) (n50) {};
    \node at (6.5,0) (n650) {};
    \node at (8,0) (n80) {};
    \node at (9,0) (n90) {};
    \node at (12,0) (n120) {};
    \node at (13,0) (n130) {};

    \node at (6.5,0.5) (n6505) {};
    
    \node at (0,1) (n01) {};
    \node at (1,1) (n11) {};
    \node at (4,1) (n41) {};
    \node at (5,1) (n51) {};
    \node at (6.5,1) (n651) {};
    \node at (8,1) (n81) {};
    \node at (9,1) (n91) {};
    \node at (12,1) (n121) {};
    \node at (13,1) (n131) {};


 \draw (n00)--(n10);
 \draw[dashed] (n10)--(n40);
 \draw (n40)--(n50);
 \draw (n50)--(n650);
 \draw (n650)--(n80);
 \draw (n80)--(n90);
 \draw[dashed] (n90)--(n120);
 \draw (n120)--(n130);
 
 \draw (n01)--(n11);
 \draw[dashed] (n11)--(n41);
 \draw (n41)--(n51);
 \draw (n51)--(n651);
 \draw (n651)--(n81);
 \draw (n81)--(n91);
 \draw[dashed] (n91)--(n121);
 \draw (n121)--(n131);
 
 \draw (n10)--(n11);
 \draw (n40)--(n41);
 \draw (n50)--(n51);
 \draw (n650)--(n6505);
 \draw (n6505)--(n651);
 \draw (n80)--(n81);
 \draw (n90)--(n91);
 \draw (n120)--(n121);

 \draw (6.5,0.5) -- (7.5,0.5);
 \draw (-0.5,0.5) circle (1cm);
 \draw (13.5,0.5) circle (1cm);
 

\end{tikzpicture}
\end{center}
\caption{The normal $7$-edge-coloring in the vertical case.} \label{fig:B1}
\end{figure}

Finally, we consider the case when $e$ is a vertical edge of the ladder. We assume the same notations that we had in the horizontal case. Now, we extend the flows $\theta_1$ and $\theta_2$ to a normal $7$-edge-coloring of $G'$ as it is shown on Figure \ref{fig:B1}.
\end{proof}

Our next theorem generalizes the result of the previous theorem for the case when we may have many pendant edges.

\begin{theorem}
\label{thm:ManyBridgesNormal7colors} Let $G'$ be a simple graph such that any of its vertices is of degree one or degree three. Moreover, assume that all bridges of $G'$ are incident to vertices of degree one. Then $\chi'_N(G')\leq 7$. 
\end{theorem}

\begin{proof} We follow the strategy of the proof of Lemma 6.3 from \cite{KraletalSteiner}. Our proof is by induction on the number of pendant edges. Clearly, we can assume that $G'$ is a connected graph. If the number of pendant edges of $G'$ is zero or one, then the statement follows from Theorem \ref{thm:7bridgeless} and Theorem \ref{thm:EndblocksNormal7colors}. Let us consider the case when this number is two. Let $u$ and $v$ be the two vertices of $G'$ that are incident to pendant edge. Consider a graph $H$ obtained from $G'$ by removing the degree-one vertices of $G'$ and adding (a possibly parallel) edge $uv$. Observe that $H$ is a bridgeless cubic graph. Hence by Theorem \ref{thm:7bridgeless} it admits a normal 7-edge-coloring. Now, consider a 7-edge-coloring of $G'$ by coloring the pendant edges of $G'$ with the color of the edge $uv$. Clearly, the coloring is normal.

Now, by induction, assume that the statement is true for all simple graphs with fewer pendant edges, and consider a simple graph $G'$ with $t\geq 3$ pendant edges. Let $u$, $v$ and $w$ be any three vertices of $G'$ incident to pendant edge. If $u$, $v$ and $w$ are pairwise adjacent, then since $G'$ is connected we have that $G'$ is obtained from a triangle by attaching a pendant edge to each of its vertices. In this case, we color $G'$ with three colors. Clearly, it is a normal 3-edge-coloring. 

Thus, without loss of generality, we can assume that $u$ and $v$ are not adjacent. Consider a graph $H$ obtained from $G'$ by removing the degree-one vertices of $G'$ incident to $u$ and $v$, and adding the edge $uv$. Observe that $H$ is a simple graph with less than $t$ pendant edges. By the induction hypothesis, it admits a normal 7-edge-coloring. Now, consider a 7-edge-coloring of $G'$ by coloring the pendant edges of $G'$ with the color of the edge $uv$. Clearly, the coloring is normal.
\end{proof}

Let $k$ be the smallest constant, such that any simple cubic graph $G$ admits a normal $k$-edge-coloring. Theorem \ref{thm:7examplethm} suggests that $k\geq 7$. A $k$-edge-coloring of a simple cubic graph is said to be strong, if any edge is rich in this coloring. In \cite{Andersen1992} Andersen has shown that any simple cubic graph admits a strong edge-coloring with ten colors. Thus, we have that $k\leq 10$. Following the approach of Andersen, in \cite{Bilkova12}, it is shown that any simple cubic graph admits a normal edge-coloring with nine colors. Thus $k\leq 9$. Now, using Theorem \ref{thm:EndblocksNormal7colors}, we further improve the latter result by obtaining the best-possible upper bound.

%

\begin{theorem}\label{thm:EnblockImplies7SimpleCase} For any simple cubic graph $G$, we have $\chi'_N(G)\leq 7$.
\end{theorem}

\begin{proof} Consider a graph $H$ obtained from $G$ by removing all the bridges of $G$. Observe that each component $C$ of $H$ is either an isolated vertex or a bridgeless graph in which all degrees are two or three. Fix a component with at least one edge. Attach to any of its degree two vertices one pendant edge such that the resulting graph meets the condition of Theorem \ref{thm:ManyBridgesNormal7colors}. We have that the resulting graph admits a normal 7-edge-coloring. Now, in order to complete the proof, observe that we can rename the colors in each component of $H$, glue the colorings in each of the components so that the resulting coloring is a normal 7-edge-coloring of $G$.
\end{proof}

\section*{Acknowledgement} We would like to thank Robert \v{S}\'{a}mal for providing us with a copy of \cite{Bilkova12} and for a useful discussion over normal colorings. The second author is indebted to Professor Cun-Quan Zhang for his advice on the topic of the present paper. Finally, we thank the two anonymous referees for their useful comments that helped us to improve the presentation of the paper.



\bibliographystyle{elsarticle-num}


\begin{thebibliography}{99}

\bibitem{Andersen1992} L. D. Andersen, The strong chromatic index of a cubic graph is at most 10, Discrete Mathematics, 108 (1992), 231--252.

\bibitem{Bilkova12} H. B\'{i}lkov\'{a}, Petersenovsk\'{e} obarven\'{i} a jeho varianty, Bachelor thesis, Charles University in Prague, Prague, 2012, (in Czech).

\bibitem{Celmins1984} A. U. Celmins, On cubic graphs that do not have an edge-$3$-colouring, Ph.D. Thesis, Department of Combinatorics and Optimization, University of Waterloo, Waterloo, Canada, 1984.



\bibitem{Fulkerson} D.R. Fulkerson, Blocking and anti-blocking pairs of polyhedra, Math. Programming 1 (1971), 168--194.

\bibitem{HaggSteff2013} J. H\"{a}gglund, E. Steffen, Petersen-colorings and some families of snarks, Ars Mathematica Contemporanea 7 (2014), 161--173.

\bibitem{HolySkoJCTB2004} F. Holyord, M. \v{S}koviera, Colouring of cubic graphs by Steiner triple systems, J. Comb. Theory, Ser. B 91, (2004), 57--66.

\bibitem{Jaeger1975} F. Jaeger, On nowhere-zero flows in multigraphs, in ``Proceedings, Fifth British Combinatorial Conference, Aberdeen, 1975." Congressus Numerantium XV, Utilitas Mathematica Winnipeg, 373--378.

\bibitem{Jaeger1979} F. Jaeger, Flows and generalized coloring theorems in graphs, J. Comb. Theory, Ser. B 26, (1979), 205--216. 

\bibitem{Jaeger1985} F. Jaeger, On five-edge-colorings of cubic graphs and nowhere-zero flow problems, Ars Combinatoria, 20-B, (1985), 229--244.

\bibitem{Jaeger1988} F. Jaeger, Nowhere-zero flow problems, Selected topics in graph theory, 3, Academic
Press, San Diego, CA, 1988, pp. 71--95.


\bibitem{KraletalSteiner} D. Kr\'{a}l', E. M\'{a}\v{c}ajov\'{a}, A. P\'{o}r, J.-S. Sereni, Characterisation Results for Steiner Triple Systems and Their Application to Edge-Colourings of Cubic Graphs, Canad. J. Math. Vol. 62 (2), 2010 pp. 355--381.



\bibitem{Mkrt2013} V. Mkrtchyan, A remark on the Petersen coloring conjecture of Jaeger, Australasian J. Comb. 56(2013), pp. 145--151.


\bibitem{Preiss1981} M. Preissmann, Sur les colorations des aretes des graphes cubiques, These de $3$-eme cycle, Grenoble (1981).

\bibitem{Samal2011} R. \v{S}\'{a}mal, New approach to Petersen coloring, Elec. Notes in Discr. Math. 38 (2011), 755--760.

\bibitem{Sch} T. Sch\"onberger, Ein Beweis des Petersenschen Graphensatzes, {\em Acta Scientia Mathematica Szeged} \textbf{7} (1934), 51--57.


\bibitem{Seymour} P. D. Seymour, On multicolourings of cubic graphs, and conjectures of Fulkerson and Tutte. Proc. London Math. Soc. 38 (3), 423--460, 1979.


\bibitem{Zhang1997} C.-Q. Zhang, Integer flows and cycle covers of graphs, Marcel Dekker, Inc., New York Basel Hong Kong, 1997.


\end{thebibliography}

\section*{Appendix}

In this section we discuss possible directions that our Lemma \ref{lem:chi7twoedgeNowhereZero8flow} can be strengthened. We present some examples which show that our lemma is best-possible.

One may wonder whether the statement of Lemma \ref{lem:chi7twoedgeNowhereZero8flow} can be strengthened to prove that the three edges of a $3$-connected cubic graph incident to a vertex can be made rich in a nowhere zero $\mathbb{Z}^{3}_{2}$-flow. This statement is not true as the following proposition shows.

\begin{proposition}
\label{prop:K33threeedges} The complete bipartite graph $K_{3,3}$ does not admit a nowhere zero $\mathbb{Z}^{3}_{2}$-flow, such that its three edges incident to a vertex are rich.
\end{proposition}

\begin{proof} Assume the opposite, and let $\theta$ be a nowhere zero $\mathbb{Z}^{3}_{2}$-flow of $K_{3,3}$ such that the three edges incident to the vertex $v$ are rich. Assume that the flow values of edges incident to $v$ are $x,y,x+y$. Consider the graph $K_{3,3}-v$, which is isomorphic to $K_{2,3}$. Observe that since the edges incident to $v$ are rich, $x,y$ and $x+y$ cannot appear on edges of $K_{3,3}-v$. Thus, there are only four non-zero values of $\mathbb{Z}^{3}_{2}$, that can appear on six edges of $K_{3,3}-v$. Hence, there are at least two edges $e_1$ and $e_2$ of $K_{3,3}-v$ which have the same flow value. Observe that $e_1$ and $e_2$ cannot be adjacent. Let the flow value of $e_1$ and $e_2$ be $z_1$, and let $z_2$ be the flow of the edge that connects $e_1$ and $e_2$. Observe that we have two edges of $K_{3,3}$ which must have flow value $z_1+z_2$. One of these edges in incident to $v$, hence $z_1+z_2\in \{x,y, x+y\}$. On the the hand, the second edge of $K_{3,3}$ with flow value $z_1+z_2$ belongs to $K_{3,3}-v$. Hence $z_1+z_2\notin \{x,y, x+y\}$. This is a contradiction.
\end{proof}

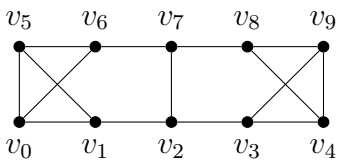
\begin{figure}[ht]

  \begin{center}

		\begin{tikzpicture}
		
		\node at (0, -0.35) {$v_0$};
		\node at (1, -0.35) {$v_1$};
		\node at (2, -0.35) {$v_2$};
		\node at (3, -0.35) {$v_3$};
		\node at (4, -0.35) {$v_4$};
		
		\node at (0, 1.35) {$v_5$};
		\node at (1, 1.35) {$v_6$};
		\node at (2, 1.35) {$v_7$};
		\node at (3, 1.35) {$v_8$};
		\node at (4, 1.35) {$v_9$};
		
		\tikzstyle{every node}=[circle, draw, fill=black!50,
                        inner sep=0pt, minimum width=4pt]

			\node[circle,fill=black,draw] at (0,0) (n00) {};
																								
			\node[circle,fill=black,draw] at (0, 1) (n01) {};
			
			\node[circle,fill=black,draw] at (1,0) (n10) {};
			\node[circle,fill=black,draw] at (1,1) (n11) {};
			
			\node[circle,fill=black,draw] at (2,0) (n20) {};
			\node[circle,fill=black,draw] at (2,1) (n21) {};
			
			\node[circle,fill=black,draw] at (3,0) (n30) {};
			\node[circle,fill=black,draw] at (3,1) (n31) {};
			
			\node[circle,fill=black,draw] at (4,0) (n40) {};
			\node[circle,fill=black,draw] at (4,1) (n41) {};

			\path[every node]
			(n00) edge  (n01)
			edge (n10)
			edge (n11)
																			
			(n01) edge (n11)
			edge (n10)
			(n10) edge (n20)
			 (n11) edge (n21)
			
			(n20) edge (n21)
			edge (n30)
			(n21) edge (n31)
			
			(n30) edge (n41)
			edge (n40)
			
			(n31) edge (n40)
			edge (n41)
			
			(n40) edge (n41)
			
			;
		\end{tikzpicture}
																
	\end{center}
	
	\caption{The vertical edge is poor in any normal $6$-edge-coloring.}\label{fig:ExampleNonrichedgeSimple6}
\end{figure}

Another question that arises is the following: can we show that any edge of a bridgeless cubic graph can be made rich in a normal $6$-edge-coloring? Our next proposition addresses this question by giving a negative answer to it.

\begin{proposition}\label{prop:Normal6coloring} The edge $v_2v_7$ is always poor in any normal $6$-edge-coloring of the graph $G$ from Figure \ref{fig:ExampleNonrichedgeSimple6}.
\end{proposition}

\begin{proof} Assume that there is a normal $6$-edge-coloring $c$ of $G$ such that $v_2v_7$ is rich. Without loss of generality, we can assume that $c(v_2v_7)=1$, $c(v_1v_2)=2$, $c(v_6v_7)=3$, $c(v_2v_3)=4$ and $c(v_7v_8)=5$. Let us show that the edge $v_0v_5$ is rich.

On the opposite assumption, assume that $v_0v_5$ is poor. Then $c(v_0v_1)=c(v_5v_6)=\alpha$ and $c(v_0v_6)=c(v_1v_5)=\beta$. Consider the edge $v_1v_5$. It has to be poor, as it is adjacent to two edges of color $\alpha$. Hence $c(v_0v_5)=2$. Now, consider the edge $v_0v_6$. Similarly, one can show that $c(v_0v_5)=3$. This gives the required contradiction. 

Thus, the edge $v_0v_5$ has to be rich, which in particular means that the colors of edges $v_0v_1$, $v_0v_6$, $v_1v_5$ and $v_5v_6$ are pairwise different. Let us show that the colors of these edges and the edge $v_0v_5$ cannot be 2 or 3. Clearly, since the graph is symmetric, we can show only for the case of color 2. Note that the edges $v_0v_1$ and $v_1v_5$ cannot have color 2. If the edge $v_0v_5$ has color 2, then the edge $v_1v_5$ has to be poor, hence the colors of edges $v_0v_1$ and $v_5v_6$ has to be the same, which gives the required contradiction. If the edge $v_5v_6$ has color 2, then the edge $v_1v_5$ has to be poor, hence the edges $v_0v_5$ and $v_0v_1$ must have the same color, which gives the required contradiction. Finally, if the color of $v_0v_6$ is 2, then the edge $v_0v_1$ has to be poor, hence the colors of edges $v_0v_5$ and $v_1v_5$ have to be the same, which gives the required contradiction. 

Thus, none of the five edges of $G$ that belong to the subgraph induced by $v_0, v_1, v_5, v_6$ can have color 2 or 3. Hence, $G$ requires at least $7$ colors in such a normal edge coloring, which in particular means that the edge $v_2v_7$ must be poor in any normal $6$-edge-coloring.
\end{proof}

Finally, one may wonder how important is the assumption of $3$-connectivity in Lemma \ref{lem:chi7twoedgeNowhereZero8flow}? Consider the graph from Figure \ref{fig:ExampleNonrichedgeSimple6}. Observe that the vertical edge is adjacent to two edges that form a 2-edge-cut. Hence for any nowhere zero $\mathbb{Z}^{3}_{2}$-flow, the values of the flow on these edges should be the same. This means that the vertical edge is going to be poor in any nowhere zero $\mathbb{Z}^{3}_{2}$-flow.

\end{document}